\documentclass{article}
 \usepackage{amsthm}
\usepackage{amsmath}
\usepackage{amssymb}
\usepackage{authblk}
\usepackage[utf8]{inputenc}
\usepackage{graphicx}
\usepackage{enumerate}
\usepackage{blkarray}
\usepackage{color}
\usepackage{subfigure}
\usepackage{a4wide}
\usepackage{graphics} 
  \usepackage{epstopdf}

\def\x{\textbf{x}}
\def\f{\bar{f}}

\def\P{\mathbb{P}}
\def\erf{\mathrm{erf}}
\def\N{\mathcal{N}}
\newtheorem{theorem}{Theorem}
\newtheorem{lemma}{Lemma}

\newtheorem{example}{Example}
\title{On the number of equilibria of the replicator-mutator dynamics for noisy social dilemmas}
\author[1]{L. Chen\thanks{lxc156@student.bham.ac.uk}}
\author[2]{C. Deng\thanks{cxd087@student.bham.ac.uk}}
\author[3]{M. H. Duong\thanks{h.duong@bham.ac.uk}}
\author[4]{T. A. Han\thanks{T.Han@tees.ac.uk}}
\affil[1,2]{Jinan-Birmingham Joint Institute, China.}
\affil[3]{School of Mathematics,
University of Birmingham, UK.}
\affil[4]{ School of Computing, Engineering and Digital Technologies, Teesside University, UK}
\date\today

\begin{document}

\maketitle
\begin{abstract}
In this paper, we consider the replicator-mutator dynamics for pairwise social dilemmas where the payoff entries are random variables. The randomness is incorporated to take into account the uncertainty, which is inevitable in practical applications and may arise from different sources such as lack of data for measuring the outcomes, noisy and rapidly changing environments, as well as unavoidable human estimate errors. We analytically and numerically compute the probability that the replicator-mutator dynamics has a given number of equilibria for four classes of pairwise social dilemmas  (Prisoner's Dilemma, Snow-Drift Game, Stag-Hunt Game and Harmony Game). 
As a result, we characterise the qualitative behaviour of such probabilities  as a function of the mutation rate. 
Our results clearly show the influence of the mutation rate   and the uncertainty in the payoff matrix definition on the number of equilibria in these games.
Overall, our analysis has provided novel theoretical contributions to the understanding of the impact of uncertainty on the behavioural diversity in a complex dynamical system.   
\end{abstract}
\section{Introduction}
Evolutionary game theory (EGT), which combines the analysis of game theory with that of dynamic evolutionary processes,  provides a powerful mathematical and simulation framework for the study of dynamics of frequencies of competing strategies in large populations \cite{maynard-smith:1982to,hofbauer:1998mm,axelrod:1981yo}. 
This framework has been successfully used for the investigation  of the evolution  of collective behaviours such as cooperation, coordination, trust and fairness, and recently, for understanding  several pressing societal challenges such as climate change and  pandemics mitigation and advanced technology governance \cite{sigmund:2010bo,perc2017statistical,han2022voluntary,SantosPNAS2011,pereira2021employing,West2007,cimpeanu2021cost}.

However,  existing works on evolutionary game theory mainly focus on deterministic games which could not capture the different random factors that  define the  interactions, in particular the game payoff matrix \cite{nowak:2006bo,sigmund:2010bo,perc2017statistical}. That is, these works assume that the outcomes of interactions for any group of strategists  (aka the game payoff matrix) can be defined in advance with certainty. However, the uncertainty in defining such outcomes is in general unavoidable which can arise from a diversity of possible   sources, including  lack of data for measuring the outcomes, noisy and rapidly changing environments, as well as unavoidable human estimate errors \cite{kahneman2021noise,may2019stability,gross2009generalized}. 

On the other hand, existing works that capture such uncertainty factors in EGT usually assume the payoff entries of the game are general random variables whose distributions are not known a prior \cite{duong2015expected,GokhalePNAS2010,han2012equilibrium,DuongHanJMB2016,DuongHanDGA2020}. In this line,  random social dilemma games where  \cite{Duong2020,Duong2021} were also analysed, where the payoff entries satisfy certain ordering required for  particular social dilemmas. These approaches are useful to provide generic properties of the underlying dynamical systems. 

However, it might be the case that  in many domains/scenarios,  
some knowledge about the payoff entries is available. In particular, some of the payoff entries might
fluctuate around certain known values, which are for example estimated through data analysis or given by domain experts. 
Capturing the available information in the analysis is essential to more accurately describe system dynamics and evolutionary outcomes.  

In this paper, we bridge this gap by analytically investigating the statistic of the number of equilibrium points in  pairwise social dilemma games where some payoff entries are drawn from a random distribution with a known mean value and variance. We will consider the full space of (symmetric) two-player games in which players can choose either to cooperate or defect (see detailed definitions in Section 1.2) \cite{santos:2006aa,sigmund:2010bo}. These games have been shown to provide important abstract frameworks to capture  collective behaviours  in a wide range of biological and social interactions such as cooperation and (anti-)coordination. 
We adopt in our analysis the replicator-mutator dynamics (see Sections 2.1 and 2.2) to model the evolutionary process, capturing both selection and mutation, allowing us to generate insights on how uncertainty in the interaction outcomes and these stochastic factors together influence equilibrium outcomes. Note that most previous works on random games consider replicator equations \cite{han2012equilibrium,duong2015expected,GokhalePNAS2010}, which is a simple version of the replicator-mutator ones where mutation is assumed to be negligible, already general enough to  encompass a variety of biological contexts from ecology to population genetics and from prebiotic to social evolution
\cite{schuster1983replicator}.  A similar setting to our work was considered in \cite{AmaralRoyal2020, AmaralPRE2020} for studying the impact of uncertainty on the evolution of cooperation, but their analysis was purely based on simulations and did not use replicator or replicator-mutator equations.    
Moreover, there has  been a particular interest in studying  average and maximal numbers of equilibrium points in a dynamical system \cite{karlin:1972aa,han2012equilibrium,altenberg2010proof}. Our analysis provides close forms for the probability of a concrete number of equilibria to occur, thus generalising these results.  


The rest of the paper is organised as follows. In Section \ref{section:mainresults} we provide the background regarding replicator-mutator dynamics and summarise the main results of the paper. Detailed proofs will follow in Sections  \ref{sec: T random}-\ref{sec: both T and S are random}.  

\section{Background and Main Results}
\label{section:mainresults}
\subsection{The replicator-mutator dynamics}
The replicator-mutator equation describes the evolution dynamics in a population of different strategies being in co-presence, where selection and mutation are both captured. It is an well established mathematical framework that integrates the unavoidable mutation observed in various  biological and social settings  \cite{traulsen:2009aa,rand2013evolution,zisis2015generosity,mcnamara2013towards,adami2016evolutionary,wang2020robust}. 
This framework  has been utilised   in many application domains, including   evolution of collective behaviours \cite{Imhof-etal2005,nowak:2006bo},  social networks dynamics \cite{Olfati2007}, language evolution
\cite{Nowaketal2001}, population genetics \cite{Hadeler1981}, and autocatalytic reaction networks
\cite{StadlerSchuster1992}. 

We consider an infinitely large population  consisting of $n$ different strategies $S_1,\cdots, S_n$. Their  frequencies are denoted, respectively,  by $x_1,\cdots, x_n$, where $\sum_{i=1}^n x_{i}=1$. These strategies undergo selection where  their  frequency, $S_i$, is determined by its fitness (i.e. average payoff), $f_i$,  obtained through interactions with others in the population. Such interactions  happen within randomly selected pairs of individuals playing a social dilemma game (see details below). 
By means of mutation, individuals in the population  might change their strategy to another randomly selected strategy, given by the so-called mutation matrix: $Q=(q_{ji}), j,i\in\{1,\cdots,n\}$. Here, $q_{ji}$ stands for  the  probability of an $S_j$ individual changing its strategy to $S_i$, satisfying that
\[
\sum_{j=1}^n q_{ji}=1, \quad \forall 1\leq i\leq n.
\]
Denoting vector  $\x = (x_1, x_2, \dots, x_n)$ and  $\f(\x)=\sum_{i=1}^n x_i f_i(\x)$  the population's average fitness,  we can describe  the replicator-mutator equation as follows \cite{Komarova2001JTB,Komarova2004,Komarova2010,Pais2012} 
\begin{equation}
\label{eq: RME}
\dot{x}_i=\sum_{j=1}^n x_j f_j(\x)q_{ji}- x_i \f(\x),\qquad  1\leq i\leq n.
\end{equation}
It is important to note that the replicator dynamics can be reproduced from \eqref{eq: RME} with $q = 0$ (i.e. no mutation). This paper  investigates the equilibrium points of the replicator-mutator dynamics, which are solutions in $[0,1]^n$ of the following system of equations
$$
g_i(\mathbf{x})=0, \quad 1\leq i\leq n,
$$
where $g_i(\x)$, $i=1,\ldots, n$, denotes the right-hand side of \eqref{eq: RME}
$$
g_i(\x):=\sum_{j=1}^n x_j f_j(\x)q_{ji}- x_i \f(\x).
$$
In general, knowing equilibrium points in a dynamical system allows us to  study states where different strategies might co-exist in the population, indicating the possibility of polymorphism.
\subsection{Pairwise social dilemmas}
Now let us consider  a pairwise  game with two strategies $S_1$ and $S_2$, with a genaral payoff matrix given below 
\begin{equation*}
\begin{blockarray}{ccc}
&S_1 & S_2 \\
\begin{block}{c(cc)}
  S_1&a_{11} & a_{12} \\
  S_2&a_{21}&a_{22}\\
  \end{block}
\end{blockarray},
\end{equation*}
where $a_{11}$ is the payoff that a player using strategy $S_1$ obtains when interacting with another player, who is also using strategy $S_1$. Other notations are interpreted similarly. 

Denoting $x$ as $S_1$'s  frequency (and thus $1-x$ as $S_2$'s  frequency),  we can simplify the replicator-mutator equation as follows 
\begin{multline}
\label{eq: 2-2 games 1}
\dot{x}=q_{11}a_{11}x^2+q_{11}x(1-x)a_{12}+q_{21}x(1-x)a_{21}+q_{21}a_{22}(1-x)^2\\-x\Big(a_{11}x^2+(a_{12}+a_{21})x(1-x)+a_{22}(1-x)^2\Big).
\end{multline}
Using the identities $q_{11}=q_{22}=1-q, \quad q_{12}=q_{21}=q$, Equation \eqref{eq: 2-2 games 1} reduces to
\begin{align}
\dot{x}&=\Big(a_{12}+a_{21}-a_{11}-a_{22}\Big)x^3+\Big(a_{11}-a_{21}-2(a_{12}-a_{22})+q(a_{22}+a_{12}-a_{11}-a_{21})\Big)x^2\nonumber
\\&\quad+\Big(a_{12}-a_{22}+q(a_{21}-a_{12}-2a_{22})\Big)x+q a_{22}.\label{eq: 2-2 games}
\end{align}
Next,  we consider a well-established parameterisation of pairwise social dilemmas, for  random game analysis   \cite{santos:2006pn,wang2015universal,szolnoki2019seasonal}. In these games, players can choose to cooperate or defect in each interaction. Mutual cooperation (punishment) would lead to a payoff $R$ ($P$). Unilateral cooperation leads to payoff $S$ while unilateral defection leads to payoff $T$. Without loss of generality, we normalize $R = 1$ and $S = 0$ in all games, and that 
$0 \leq a_{21} = T \leq 2$ and $-1 \leq a_{12} = S \leq 1$. We focus on four important social dilemma games, characterised by different orderings of the payoff entries 
\begin{enumerate}[(i)]
\item the Prisoner's Dilemma (PD): $2\geq T > 1 > 0 > S\geq -1$ (both players defect),
\item the Snow-Drift (SD) game: $2\geq T > 1 > S > 0$ (players prefer unilateral defection to mutual cooperation),
\item the Stag Hunt (SH) game: $1 > T  > 0 > S\geq -1$ (players prefer mutual defection to unilateral cooperation),
\item the Harmony (H) game: $1  > T\geq 0, 1\geq S > 0$ (both players cooperate).
\end{enumerate}
As motivated in the introduction, to be more realistic and to capture the various possible uncertainty, we will consider random games, where $T$ or $S$ or both T and S, are random variables. 
\subsection{Main results}
We obtain explicit analytical formulas for the probability that the replicator-mutator dynamics has a certain number of equilibria for the four social dilemmas above in two different cases, where $T$ or $S$ is random and the other is fixed. The distinction is necessary since $T$ and $S$ might play different roles in the equilibrium outcomes, which can be also seen from our results.
\begin{theorem}[$T$ is random] Suppose that $T$ is normally distributed with mean $T_0$ and variance $\sigma^2$, i.e. $T\sim\N(T_0,\sigma^2)$, and $S$ is a given number where $T_0$ and $S$ satisfy the corresponding ordering in each of the  social dilemmas above. Let $T_1, T_2,T_3$ and $s_1, s_2$ be defined in \eqref{eq: T1}, \eqref{eq: T2} and \eqref{eq: s1s2T3}. Then, for all games, the probability that the replicator-mutator has 2 equilibria is given by
$$
p_2= \dfrac{1}{2}-\dfrac{1}{2}\erf\Big(\dfrac{T_3-T_0}{\sigma\sqrt{2}}\Big).
$$
The probability that the replicator-mutator has three equilibria is given as follows.  First,  for SD and H games
$$
p_3=1-p_2.
$$
Now, for PD and SH games
\begin{equation*}
    p_3=\begin{cases}
    \dfrac{1}{2}\left[1+\erf\Big(\dfrac{T_1-T_0}{\sigma\sqrt{2}}\Big)\right], \quad\text{if}\quad  S \leq s_1,\\\\
    \dfrac{1}{2}\left[1+\erf\Big(\dfrac{T_1-T_0}{\sigma\sqrt{2}}\Big)+\erf\Big(\dfrac{T_3-T_0}{\sigma\sqrt{2}}\Big)-\erf\Big(\dfrac{T_2-T_0}{\sigma\sqrt{2}}\Big)\right], \quad\text{if}\quad  s_1<S<s_2,\\\\
     \dfrac{1}{2}\left[1+\erf\Big(\dfrac{T_3-T_0}{\sigma\sqrt{2}}\Big)\right], \quad\text{if}\quad S \geq s_2.
    \end{cases}
\end{equation*}
Thus the probability that the replicator-mutator has one equilibrium is, for SD and H games $p_1=0$, and for PD and SH games, $p_1=1-p_2-p_3$.

Furthermore, as a function of $q$, the probability $p_2$ is decreasing in SD and H games, but is increasing in PD and SH games and satisfies the following small mutation limit 
$$
\lim_{q\rightarrow 0}p_2=\begin{cases}
1, \quad \text{in SD and H games},\\
0, \quad \text{in PD and SH games}.
\end{cases}
$$
In addition, in SD and H games, it holds that
\begin{equation*}
p_2\geq \frac{1}{2}-\frac{1}{2}\erf\Big(-\frac{S+T_0}{\sigma\sqrt{2}}\Big)>\frac{1}{2}.
\end{equation*}
\end{theorem}
\begin{theorem}[$S$ is random]Suppose that $S$ is normally distributed with mean $S_0$ and variance $\sigma^2$, i.e. $S\sim\N(S_0,\sigma^2)$, and $T_0$ is a given number so that $T_0$ and $S$ satisfy the corresponding ordering in each of the  social dilemmas above. Let $S_1$, $S_2$ be defined in \eqref{S1S2}. The probabilities that the replicator-mutator dynamics has $1$,$2$ and $3$ equilibria are given by, respectively,
\begin{align*}
p_1&=1-p_2-p_3,\\
p_2&=\frac{1}{2}-\frac{1}{2}\erf(-\frac{\frac{qT_0}{1-q}-S_0}{\sqrt{2}\eta}),\\
p_3&=
\left\{
   \begin{aligned}
    &\frac{1}{2}+\frac{1}{2}\erf(\frac{S_1-S_0}{\sqrt{2}\eta}), &\text{if } &T_0>\frac{(1-q)^2}{1-2q},\\
    &\frac{1}{2}+\frac{1}{2}\erf(\frac{S_1-S_0}{\sqrt{2}\eta})+\frac{1}{2}\erf(\frac{\frac{-qT_0}{1-q}-S_0}{\sqrt(2}\eta)-\frac{1}{2}\erf(\frac{S_2-S_0}{\sqrt{2}\eta}), &\text{if } &T_0<\frac{(1-q)^2}{1-2q}.
   \end{aligned} 
\right..
  \end{align*}
 As a consequence, $p_2$ is always increasing as a function of $q$ and satisfies the following lower and upper bounds
 $$
\frac{1}{2}-\frac{1}{2}\erf(\frac{-S_0}{\sqrt{2}\eta})< p_2< \frac{1}{2}-\frac{1}{2}\erf(\frac{-T_0-S_0}{\sqrt{2}\eta}).
$$
\end{theorem}
\vspace{0.5cm}
In the following we will provide detailed proofs for these main results. In Section \ref{sec: T random} we study the case where $T$ is random and $S$ is deterministic. We compute the probability that the replicator-mutator dynamics has $p_k$ ($k\in\{1,2,3\}$) equilibria, both analytically and numerically by   sampling the payoff matrix space. Similar results where $S$ is random and $T$ is deterministic are obtained in Section \ref{sec: S random}. In Section \ref{sec: both T and S are random} we numerically investigate the case where both $T$ and $S$ are random. Summary and outlook is given in the final section,  Section \ref{sec: summary}.

\section{$T$ is random}
\label{sec: T random}
We first consider the case where only $T$ is random. Specifically, we assume that
\begin{equation}
    T=T_0+\varepsilon_T,
\end{equation}
where $\varepsilon_T$ is a centered random variable. Suppose that $\varepsilon_T$ is a centered normal distribution, $\varepsilon_T \sim \N(0,\sigma^2)$, and $T_0$ is a fixed number.  It follows that
\[T \sim \N(T_0,\sigma^2).\]

This means that we have partial information about the value of $T$, which is randomly fluctuating (perturbed) around a deterministic value $T_0$. In practical applications, this may come from estimations based on expert's advice or data simulations. By taking the value of the variance smaller and smaller, the value of $T$ is more and more concentrated around $T_0$, and in particular, by sending the variance to zero, we expect to recover deterministic games.  
\subsection{Equilibrium points}
\label{sec: equilibrium points}
By simplifying the right hand side of \eqref{eq: 2-2 games}, equilibria of a social dilemma game are roots in the interval $[
0,1]$ of the following cubic equation 
\begin{align}
\Big(T+S-1\Big)x^3+\Big(1-T-2S+q(S-1-T)\Big)x^2+\Big(S+q(T-S)\Big)x = 0.\label{eq: 2-2 games2}
\end{align}
It follows that $x = 0$ is always an equilibrium. If $q=0$, \eqref{eq: 2-2 games2} reduces to
$$
(T+S-1)x^3+(1-T-2S)x^2+Sx=0,
$$
which has solutions
$$
x=0, \quad x=1, \quad x^*=\frac{S}{S+T-1}.
$$
Note that for SH and SD games $x^*\in (0,1)$, thus it is always an (internal) equilibrium. On the other hand, for PD-games and H-games, $x^* \not\in (0,1)$, thus it is not an equilibrium.

If $q=\frac{1}{2}$ then the above equation has two solutions $x_1=\frac{1}{2}$ and $x_2=\frac{T+S}{T+S-1}$. In PD, SD and H games, $x_2\not \in (0,1)$, thus they have two equilibria $x_0=0$ and $x_1=\frac{1}{2}$. In the SH game: if $T+S<0$ then the game has three equilibria $x_0=0, x_1=\frac{1}{2}$ and $0<x_2<1$; if $T+S\geq 0$ then the game has only two equilibria $x_0=0, x_1=\frac{1}{2}$.

Now we consider the case $0<q<\frac{1}{2}$. For non-zero equilibrium points we solve the following quadratic equation
\begin{equation}
\label{eq: social dilemma}
h(x):=(T+S-1)x^2+(1-T-2S+q(S-1-T))x+S+q(T-S)=:ax^2+bx+c=0,
\end{equation}
where we define
\begin{equation}
\label{eq: abc}    
a=T+S-1, \quad b=1-T-2S+q(S-1-T), \quad c= S+q(T-S).
\end{equation}
Set $t:=\frac{x}{1-x}$, then we obtain 
$$
\frac{h(x)}{(1-x)^2}=(a+b+c) t^2+(b+2c)t+c=-qt^2+(-q-a+c)t+c:=g(t).
$$
Thus, an equilibrium point of a social dilemma can be found from a positive solution of the following quadratic equation
\begin{equation}
\label{eq: g}    
g(t)=-qt^2+(-q-a+c)t+c=-q t^2 - ((1-q)(T-1)+qS)t + S+q(T-S).
\end{equation}
Using this relation, subsequently we provide numerical simulations and analytical results for the probability $p_k$ that each of the game mentioned above has a certain number, $k\in \{1,2,3\}$, of equilibria. For the quadratic function $g(t)$, the discriminant is given by
\begin{equation}
    \Delta=(q+a-c)^2+4qc,
\end{equation}
where $a$ and $c$ are defined in terms of $T$ and $S$ in \eqref{eq: abc}. The number of positive roots of $g$ is characterized in the following three cases:
\begin{enumerate}[(a)]
    \item  $g$ has no positive roots, which happens when $g$ has no real roots ($\Delta<0$) or $g$ has only negative roots ( $\Delta\geq0, t_1\leq0, t_2\leq0$). In this case, the replicator-mutator equation has only one equilibrium $x=0$.
    \item $g$ has one positive root, which happens when $g$ has a positive double root ($\Delta=0, t_1=t_2>0$) or  when $g$ has one positive and one negative root ($\Delta>0, t_1 t_2<0$). In this case, the replicator-mutator equation has two equilibria.
    \item $g$ has two positive roots ($\Delta>0, t_1+t_2>0, t_1 t_2>0$), thus the replicator-mutator equation has three equilibria.
\end{enumerate}
Using \eqref{eq: abc} we can write
\begin{eqnarray}
\Delta & = & (q+a-c)^2+4qc \nonumber \\
~ & = & [q+T+S-1-(S+qT-qS)]^2+4q[S+q(T-S)] \nonumber \\
~ & = & (1-q)^2T^2+2[(q-q^2)S+q^2+2q-1]T+q^2S^2+2q(1-q)S+(q-1)^2, \nonumber
\end{eqnarray}
which is a quadratic function of $T$ whose discriminant is given by
\begin{equation*}
    2q^3-q^2-(2q^3-3q^2+q)S=q(2q-1)[q-(q-1)S].
\end{equation*}
Therefore
\begin{enumerate}[(i)]
    \item If $2q^3-q^2-(2q^3-3q^2+q)S>0$, then $\Delta=0$ has two distinct solutions.
    \item If $2q^3-q^2-(2q^3-3q^2+q)S=0$, then $\Delta=0$ has only one solution.
    \item If $2q^3-q^2-(2q^3-3q^2+q)S<0$, then $\Delta=0$ has no real solutions.
\end{enumerate}
\noindent If $(q-1)S\leq q$, then $\Delta=0$ has solutions
\begin{align}
T_1&=\frac{-((1-q)qS+q^2+2q-1)-2\sqrt{2q^3-q^2-(2q^3-3q^2+q)S}}{(1-q)^2},\label{eq: T1}\\
T_2&=\frac{-((1-q)qS+q^2+2q-1)+2\sqrt{2q^3-q^2-(2q^3-3q^2+q)S}}{(1-q)^2}.\label{eq: T2}
\end{align}
\noindent In addition, we also have the following:
\begin{align*}
t_1+t_2&=  \dfrac{q+a-c}{-q}=  \dfrac{q+T+S-1-(S+qT-qS)}{-q}  =  \dfrac{q+T-qT+qS-1}{-q},\\
t_1t_2 &= \dfrac{c}{-q} = \dfrac{S+qT-qS}{-q}.
\end{align*}
 \noindent Define
\begin{equation}
\label{eq: s1s2T3}    
s_1=\dfrac{-q(1-q)}{1-2q}, \quad s_2=\dfrac{q}{q-1},\quad \text{and}\quad T_3=\dfrac{(q-1)S}{q}.
\end{equation}
Note that since $0<q<1/2$, $s_1\in(-\infty,0)$, $s_2\in(-1,0)$ and $s_1<s_2$. Furthermore,the condition that $(q-1)S\leq q$ will always hold  when $S>0$, that is for SD and H games. 
\subsection{Probability that the replicator-mutator has two equilibria}
We first compute the probability $p_2$ that the replicator-mutator equation has two equilibria, which amounts to compute the probability that $\Delta=0, t_1=t_2>0$ or $\Delta>0, t_1t_2<0$:
$$
p_2=\mathbb{P}\Big\{\Delta=0, t_1=t_2>0\}+\mathbb{P}\Big\{\Delta>0, t_1t_2<0\}=p_{21}+p_{22}.
$$
\noindent We note that when $\Delta=0$ has no solution, $\mathbb{P}\{\Delta=0, t_1=t_2>0\}=0$. When $\Delta=0$ has two solutions $T_1$ and $T_2$ ($T_1$ may equal to $T_2$ here, but it does not matter), then $\mathbb{P}\{\Delta=0, t_1=t_2>0\}$ is given by
\[P(T=T_1)+P(T=T_2).\] 
\noindent Since in both cases (Gaussian and uniform distributions) $T$ follows a continuous distribution, we have $P(T=T_1)=P(T=T_2)=0$. It follows that $p_{21}=0$. Thus 
$$
p_2=p_{22}=\mathbb{P}\{\Delta>0, t_1t_2<0\},
$$
which is equivalent to the probability that $T>T_3$. 
Since $T\sim \mathcal{N}(T_0,\sigma^2)$, we obtain
\begin{align}
p_2 & = \mathbb{P}(T>T_3) \label{eq: p2T}\\
&= \int_{T_3}^{\infty} \dfrac{1}{\sigma\sqrt{2\pi}}\exp\Big({\dfrac{-(x-T_0)^2}{2\sigma^2}}\Big) dx\notag 
\\&=\dfrac{1}{\sqrt{\pi}}\int_{\frac{T_3-T_0}{\sigma\sqrt{2}}}^{\infty} \exp(-t^2) dt\quad (\text{by changing of variable}~ t=\frac{x-T_0}{\sigma\sqrt{2}})\notag
\\&=\dfrac{1}{\sqrt{\pi}}\left[\int_{0}^{\infty} \exp(-t^2) dt-\int_0^{\frac{T_3-T_0}{\sigma\sqrt{2}}} \exp(-t^2) dt\right]\notag\\
& =  \dfrac{1}{2}-\dfrac{1}{2}\erf\Big(\dfrac{T_3-T_0}{\sigma\sqrt{2}}\Big)\notag
\\&= \dfrac{1}{2}-\dfrac{1}{2}\erf\Big(\dfrac{(q-1)S/q-T_0}{\sigma\sqrt{2}}\Big),\notag
\end{align}
where $\erf(\cdot)$ is the error function
$$
\mathrm{erf}(z):=\frac{2}{\sqrt{\pi}}\int_0^z e^{-t^2}\,dt.
$$
The following lemma presents some interesting qualitative properties of $p_2$.
\begin{lemma}
\label{lem: properties of p2}
As a function of $q$, the probability $p_2$ is decreasing in SD and H games, but is increasing in PD and SH games. As a consequence, for SD and H games, it holds that
\begin{equation}
\label{eq: lowerbound p2}
p_2\geq \frac{1}{2}-\frac{1}{2}\erf\Big(-\frac{S+T_0}{\sigma\sqrt{2}}\Big)>\frac{1}{2}.
\end{equation}
In addition,
$$
\lim_{q\rightarrow 0}p_2=\begin{cases}
1, \quad \text{in SD and H games},\\
0, \quad \text{in PD and SH games}.
\end{cases}
$$
\end{lemma}
It is worth mentioning that the monotonicity and the small mutation limits above are independent of the specific values of $S,T_0$ and  $\sigma$. The lower bound \eqref{eq: lowerbound p2} indicates that in SD and H games, $p_2$ is always dominant over $p_1$ and $p_3$ since $p_1+p_2+p_3=1$. For instance, in the SD game, with the specific values $S=0.5, T_0=1.5, \sigma=1$ we obtain 
$$
p_2\geq \frac{1}{2}-\frac{1}{2}\erf(-\sqrt{2})\approx 0.97725,
$$
and in the H game with the specific values $S=0.5, T=0.5, \sigma=1$, we obtain 
$$
p_2\geq \frac{1}{2}-\frac{1}{2}\erf(-1/\sqrt{2})\approx 0.8413.
$$
\begin{proof}
Let $z=\dfrac{(q-1)S/q-T_0}{\sigma\sqrt{2}}$, then
$$
p_2=\frac{1}{2}-\frac{1}{2}\erf{(z)}.
$$
Therefore, using the chain rule and the fact that $\frac{d}{dz}\erf(z)=\frac{2}{\sqrt{\pi}} e^{-z^2}$, we obtain
$$
\frac{dp_2}{dq}=\frac{dp_2}{dz}\frac{dz}{dq}=-\frac{1}{\sqrt{\pi}} e^{-z^2}\frac{S}{\sigma q^2\sqrt{2}}.
$$
Since $S$ is positive in SD and H games but is negative in PD and SH games, $\frac{dp_2}{dq}$ is negative in SD and H games but is positive in PD and SH games. Thus the probability $p_2$, as a function of $q$, is decreasing in SD and H games, but is increasing in PD and SH games. As a consequence, for SD and H games, we have
$$
p_2\geq p_2\vert_{q=1/2}=\frac{1}{2}-\frac{1}{2}\erf\Big(-\frac{S+T_0}{\sigma\sqrt{2}}\Big).
$$
Now we establish the limit of $p_2$ as $q$ tends to $0$. Since $0<q<\frac{1}{2}$ we have
$$
\lim_{q\rightarrow 0}z=\begin{cases}
-\infty,\quad\text{if}~~S>0,\\
+\infty,\quad\text{if}~~S<0.
\end{cases}
$$
Together with the fact that
$$
\lim_{z\rightarrow \pm \infty}\erf(z)=\pm 1,
$$
we obtain
$$
\lim_{q\rightarrow 0} p_2=\frac{1}{2}-\frac{1}{2}\lim_{q\rightarrow 0}\erf(z)=\begin{cases}
1,\quad\text{if}~~S>0,\\
0,\quad\text{if}~~S<0.
\end{cases}
$$
Applying this to the underlying games, we achieve  
$$
\lim_{q\rightarrow 0}p_2=\begin{cases}
1, \quad \text{in SD and H games},\\
0, \quad \text{in PD and SH games}.
\end{cases}
$$
\end{proof}

\subsection{Probability that the replicator-mutator has three equilibria}
Now we compute the probability $p_3$ that the replicator-mutator equation has three equilibria. We consider the following cases, depending on the ordering between $S$ and $s_1< s_2$.
\subsubsection*{(C1) $S>q/(q-1)=s_2$} Then, recalling that when $0<q<1/2$, then $2q^3-q^2-(2q^3-3q^2+q)S=q(2q-1)[q-(q-1)S]<0$. Therefore, $\Delta>0$ always holds. It also holds that $T_3<\dfrac{1-(S+1)q}{1-q}$. Therefore,
$$
\Big\{T: \Delta  >0, T  <\dfrac{1-(S+1)q}{1-q},T <T_3\Big\}=\{T: T<T_3\}.
$$
Thus, when $(q-1)S>q$, the probability $p_3$ that the replicator-mutator dynamics has three equilibria is given by
\begin{align*}
p_3 & = \P(T<T_3)=1-\P(T>T_3)=1-p_2
\\&=\frac{1}{2}+\frac{1}{2}\erf\Big(\frac{T_3-T_0}{\sigma\sqrt{2}}\Big).
\end{align*}
It implies that in this case, $p_1=0$. Note that, SD and H games always satisfy this case since $s_2=\frac{q}{q-1}<0<S$. Thus, we obtain for SD and H games:
$$
p_1=0,\quad p_2=\dfrac{1}{2}-\dfrac{1}{2}\erf\Big(\dfrac{(q-1)S/q-T_0}{\sigma\sqrt{2}}\Big),\quad p_3=1-p_2=\dfrac{1}{2}+\dfrac{1}{2}\erf\Big(\dfrac{(q-1)S/q-T_0}{\sigma\sqrt{2}}\Big).
$$

\subsubsection*{(C2) $S=s_2$}
Then $2q^3-q^2-(2q^3-3q^2+q)S=0$, we obtain  that
\begin{equation}
    T\neq-\dfrac{q(1-q)S+q^2+2q-1}{(1-q)^2} \Leftrightarrow \Delta>0.
\end{equation}
\noindent A direct computation shows that $-\dfrac{q(1-q)S+q^2+2q-1}{(1-q)^2}<T_3$. Thus, the probability that the replicator-mutator equation has 3 equilibria is again $p_3=\P(T<T_3)$. Thus we obtain the same results as in the previous case.
\subsubsection*{(C3) $s_1<S<s_2$}
\noindent In this case since $2q^3-q^2-(2q^3-3q^2+q)S>0$, $\Delta=0$ has two distinct solutions $T_1$ and $T_2$. In addition, we have
\begin{equation*}
    \dfrac{1-(S+1)q}{1-q}>T_3>T_2.
\end{equation*}
It implies that the probability of the replicator-mutator equation having three equilibria is
\begin{align*}
 p_3 & = \P(T<T_1)+\P(T_2<T<T_3) \nonumber \\
&=  \int_{-\infty}^{T_1} \dfrac{1}{\sigma\sqrt{2\pi}}\exp\Big({\dfrac{-(x-T_0)^2}{2\sigma^2}}\Big) dx + \int_{T_2}^{T_3} \dfrac{1}{\sigma\sqrt{2\pi}}\exp\Big({\dfrac{-(x-T_0)^2}{2\sigma^2}}\Big) dx   
\\&=\frac{1}{2}\left[1+\erf\Big(\frac{T_1-T_0}{\sigma\sqrt{2}}\Big)+\erf\Big(\frac{T_3-T_0}{\sigma\sqrt{2}}\Big)-\erf\Big(\frac{T_2-T_0}{\sigma\sqrt{2}}\Big)\right].
\end{align*}
\subsubsection*{(C4) $S \leq s_1$} Similarly to  case $(C3)$, $\Delta=0$ has two distinct solutions $T_1$ and $T_2$ and
\begin{equation*}
    \dfrac{1-(S+1)q}{1-q} \leq T_3,\quad   T_1 \leq \dfrac{1-(S+1)q}{1-q} \leq T_2.
\end{equation*}
Thus the probability that replicator-mutator equation has three equilibria is given by
\begin{align*}
p_3 & = \P(T<T_1) 
\\& = \int_{-\infty}^{T_1} \dfrac{1}{\sigma\sqrt{2\pi}}\exp\Big({\dfrac{-(x-T_0)^2}{2\sigma^2}}\Big) dx
\\&=\frac{1}{2}+\frac{1}{2}\erf\Big(\frac{T_1-T_0}{\sigma\sqrt{2}}\Big).
\end{align*}
Bringing all cases together, we obtain, in PD and SH games, the probability that the replicator-mutator equation has three equilibria is
\begin{equation*}
    p_3=\begin{cases}
    \dfrac{1}{2}\left[1+\erf\Big(\dfrac{T_1-T_0}{\sigma\sqrt{2}}\Big)\right], \quad\text{if}\quad  S \leq s_1,\\\\
    \dfrac{1}{2}\left[1+\erf\Big(\dfrac{T_1-T_0}{\sigma\sqrt{2}}\Big)+\erf\Big(\dfrac{T_3-T_0}{\sigma\sqrt{2}}\Big)-\erf\Big(\dfrac{T_2-T_0}{\sigma\sqrt{2}}\Big)\right], \quad\text{if}\quad  s_1<S<s_2,\\\\
     \dfrac{1}{2}\left[1+\erf\Big(\dfrac{T_3-T_0}{\sigma\sqrt{2}}\Big)\right], \quad\text{if}\quad S \geq s_2.
    \end{cases}
\end{equation*}
The three cases in the  formula above can be written in terms of $q$ as follows
\begin{equation*}
    p_3=\begin{cases}
    \dfrac{1}{2}\left[1+\erf\Big(\dfrac{T_1-T_0}{\sigma\sqrt{2}}\Big)\right], \quad\text{if}\quad  0\leq q\leq q_1,\\\\
    \dfrac{1}{2}\left[1+\erf\Big(\dfrac{T_1-T_0}{\sigma\sqrt{2}}\Big)+\erf\Big(\dfrac{T_3-T_0}{\sigma\sqrt{2}}\Big)-\erf\Big(\dfrac{T_2-T_0}{\sigma\sqrt{2}}\Big)\right], \quad\text{if}\quad  q_1<q<q_2,\\\\
     \dfrac{1}{2}\left[1+\erf\Big(\dfrac{T_3-T_0}{\sigma\sqrt{2}}\Big)\right], \quad\text{if}\quad q_2 \leq q\leq \frac{1}{2},
    \end{cases}
\end{equation*}
where $q_1$ and $q_2$ are respectively unique solutions of
$$
\frac{q(q-1)}{1-2q}=S\quad \text{and}\quad \frac{q}{q-1}=S.
$$
We note that when $q_2\leq q\leq \frac{1}{2}$ then $p_3=1-p_2$, thus $p_1=0$.
Using the above formula, in principle, we can derive qualitative properties for $p_3$ as a function of $q$ as in Lemma \ref{lem: properties of p2}. However, to compute the derivatives of $T_1$ and $T_2$ with respect to $q$ and determine their signs for general $S$ are very complicated for general $S$. In the following example, we demonstrate such a result for the specific value $S=-0.5$.
\begin{example}
For $S=-0.5$, we obtain the following formula for $p_3$ depending on the value of $q$
\begin{equation*}
    p_3=\begin{cases}
    \dfrac{1}{2}\left[1+\erf\Big(\dfrac{T_1-T_0}{\sigma\sqrt{2}}\Big)\right], \quad\text{if}\quad  q \leq \frac{2-\sqrt{2}}{2},\\\\
    \dfrac{1}{2}\left[1+\erf\Big(\dfrac{T_1-T_0}{\sigma\sqrt{2}}\Big)+\erf\Big(\dfrac{T_3-T_0}{\sigma\sqrt{2}}\Big)-\erf\Big(\dfrac{T_2-T_0}{\sigma\sqrt{2}}\Big)\right], \quad\text{if}\quad \frac{2-\sqrt{2}}{2}\leq q\leq \frac{1}{3} ,\\\\
     \dfrac{1}{2}\left[1+\erf\Big(\dfrac{T_3-T_0}{\sigma\sqrt{2}}\Big)\right], \quad\text{if}\quad \frac{1}{3}\leq q\leq \frac{1}{2}.
    \end{cases}
\end{equation*}
From the above analytical formula, we can study the beheviour of $p_3$ as a function of $q$. For $0<q<\frac{2-\sqrt{2}}{2}$, we have
$$
\frac{dp_3}{dq}=\frac{1}{\sigma\sqrt{2\pi}}\exp\Big[-\frac{(T_1-T_0)^2}{2\sigma^2}\Big]\frac{dT_1}{dq}<0.
$$
For $\frac{2-\sqrt{2}}{2}<q<\frac{1}{3}$,
$$
\frac{dp_3}{dq}=\frac{1}{\sigma\sqrt{2\pi}}\bigg(\exp\Big[-\frac{(T_1-T_0)^2}{2\sigma^2}\Big]\frac{dT_1}{dq}+\exp\Big[-\frac{(T_3-T_0)^2}{2\sigma^2}\Big]\frac{dT_3}{dq}-\exp\Big[-\frac{(T_2-T_0)^2}{2\sigma^2}\Big]\frac{dT_2}{dq}\bigg)>0.
$$
For $\frac{1}{3}<q<\frac{1}{2}$,
$$
\frac{dp_3}{dq}=\frac{1}{\sigma\sqrt{2\pi}}\exp\Big[-\frac{(T_3-T_0)^2}{2\sigma^2}\Big]\frac{dT_3}{dq}<0
$$
Thus as a function of $q$, $p_3$ is decreasing in $(0,\frac{2-\sqrt{2}}{2})$, is increasing in $(\frac{2-\sqrt{2}}{2}, \frac{1}{3})$ and then is decreasing again in $(\frac{1}{3},\frac{1}{2})$.
\end{example}
 \begin{figure}
\centering
\subfigure[PD]{
    \includegraphics[scale=1]{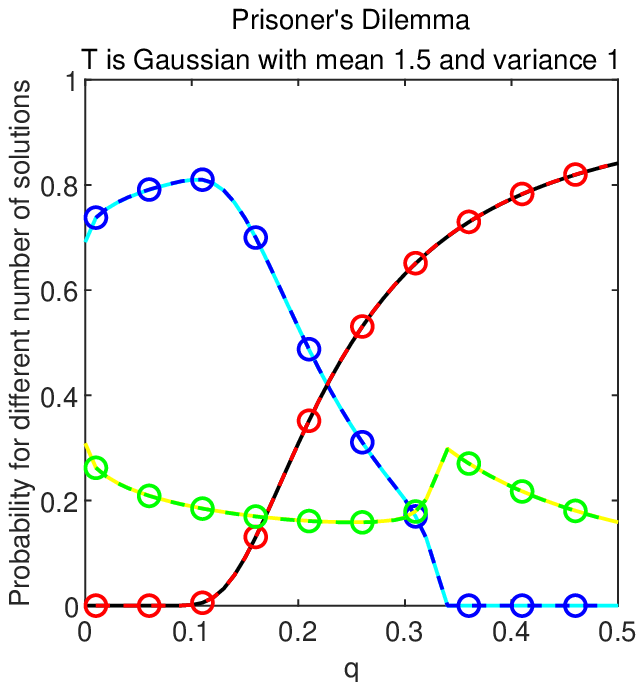}
    }
    \quad
    \subfigure[SD]{
     \includegraphics[scale=1]{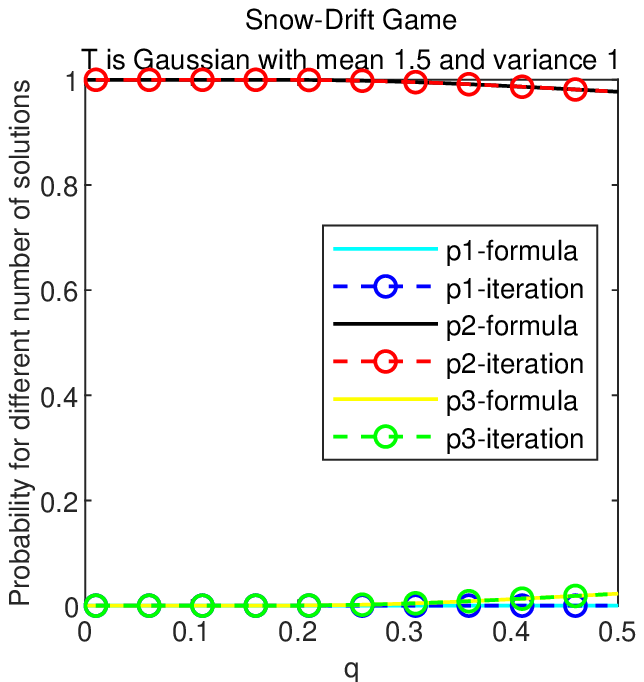}
    }
    \quad
    \subfigure[SH]{
    \includegraphics[scale=1]{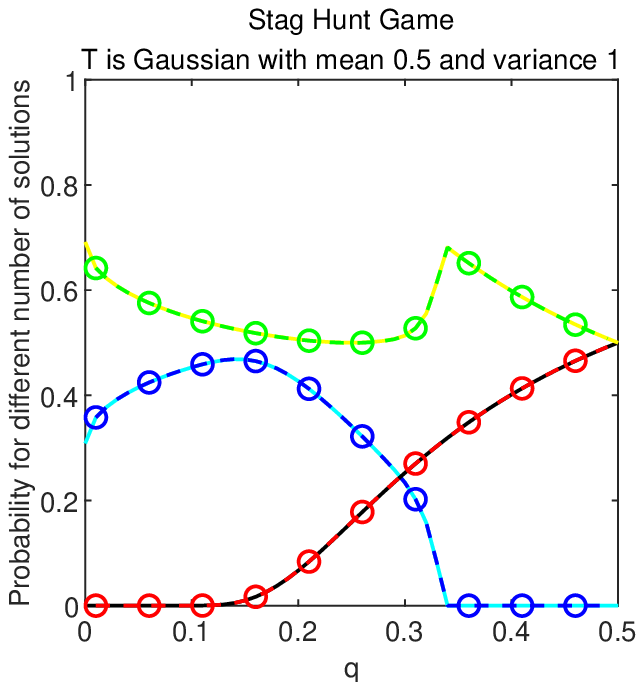}
    }
    \quad
    \subfigure[H]{
     \includegraphics[scale=1]{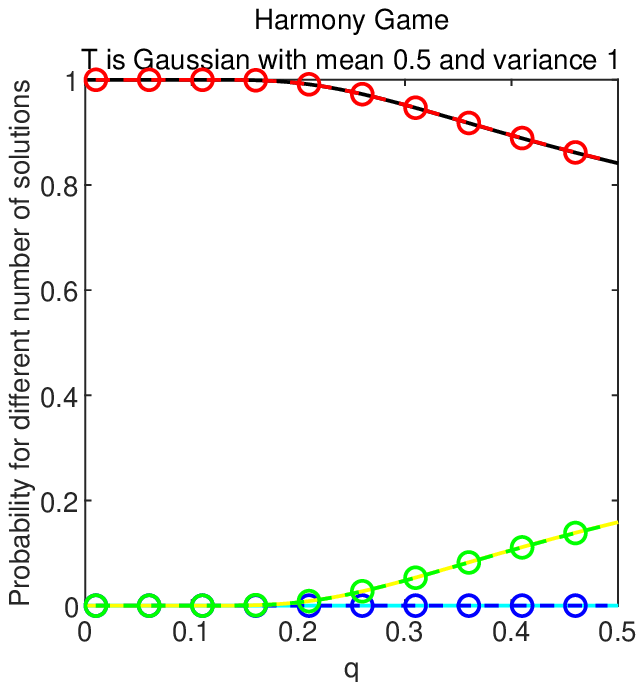}
    }
\caption{Probabilities $p_1, p_2 ,p_3$ that the replicator-mutator dynamics has respectively $1,2, 3$ equilibria as functions of the mutation rate $q$ in PD, SD, SH and H games when $T$ is random and $S$ is fixed. The values from analytical analysis and numerical samplings are in accordance.} 
\label{fig: fig3}
\end{figure}

\subsection{Numerical simulations}
\label{sec: Trandom numerics}
In Figure \ref{fig: fig3} we plot the probabilities that PD, SD, SH and H games have $1, 2$ or $3$ equilibria for different values of $q$ using the analytical formula obtained in the previous section (see Theorem 1). For validation, these probabilities are also computed by sampling over $10^6$ realizations of $T$ from the normal distribution $\N(T_0,1)$ and then calculating the solutions of $g(t)=0$. In these simulations, the value $S_0$ (respectively, $T_0$) is taken to be the middle point in the interval that $S$ (respectively, $T$) belongs to in each corresponding game, that is $S_0=-0.5$ in PD and SH games  and $S_0=0.5$ in SD and H games (respectively, $T_0=1.5$ in PD and SD games, $T_0=0.5$ in SH and H games).
It can be clearly seen that the simulation results are in accordance with the analytical ones. In particular, 
we observe that, as a function of $q$, the probability $p_2$ is decreasing in SD and H games, and is increasing in PD and SH games.


\label{sec: Trandom sigma numerics}
In Figure \ref{fig: fig6} we plot the probabilities that PD, SD, SH and H games have $1, 2$ or $3$ equilibria for different values of $\sigma$ (fixing $q=0.25$), using the analytical formula obtained in the previous section. The value $S_0$ (respectively, $T_0$) is taken to be the middle point in the interval that $S$ (respectively, $T$) belongs to in each corresponding game, that is $S_0=-0.5$ in PD and SH games  and $S_0=0.5$ in SD and H games (respectively, $T_0=1.5$ in PD and SD games, $T_0=0.5$ in SH and H games).

\begin{figure}[htbp]
\centering
\subfigure[PD]{
    \includegraphics[scale=1]{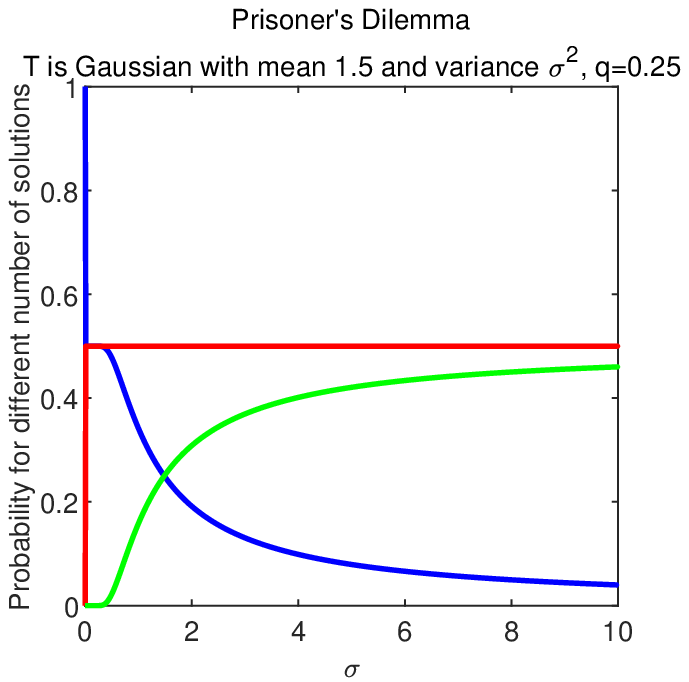}
    }
    \quad
    \subfigure[SD]{
     \includegraphics[scale=1]{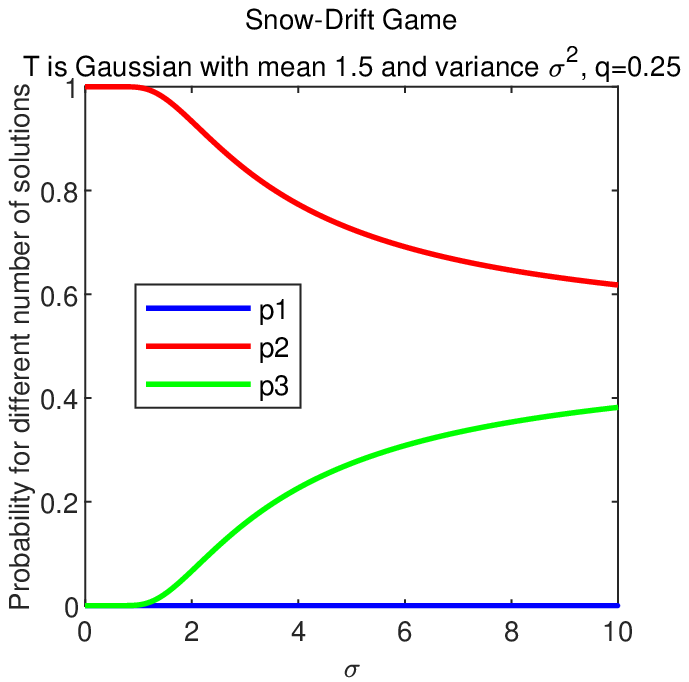}
    }
    \quad
    \subfigure[SH]{
    \includegraphics[scale=1]{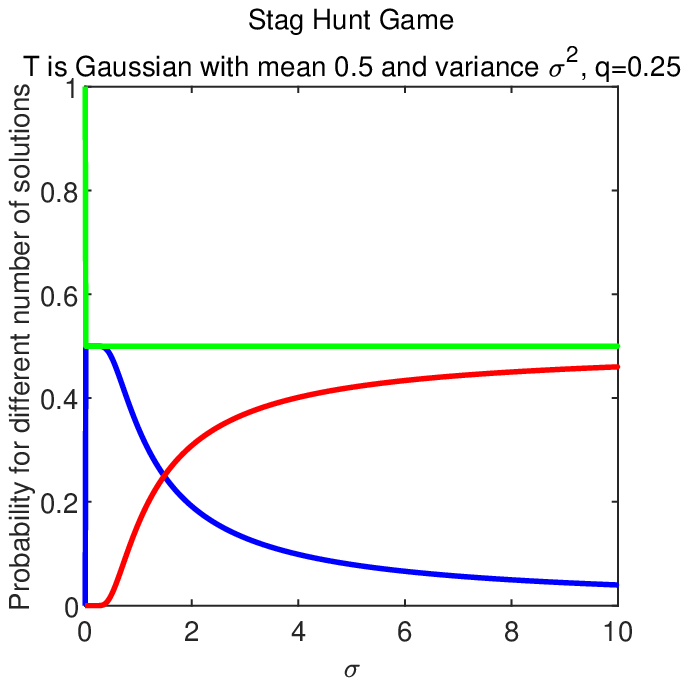}
    }
    \quad
    \subfigure[H]{
     \includegraphics[scale=1]{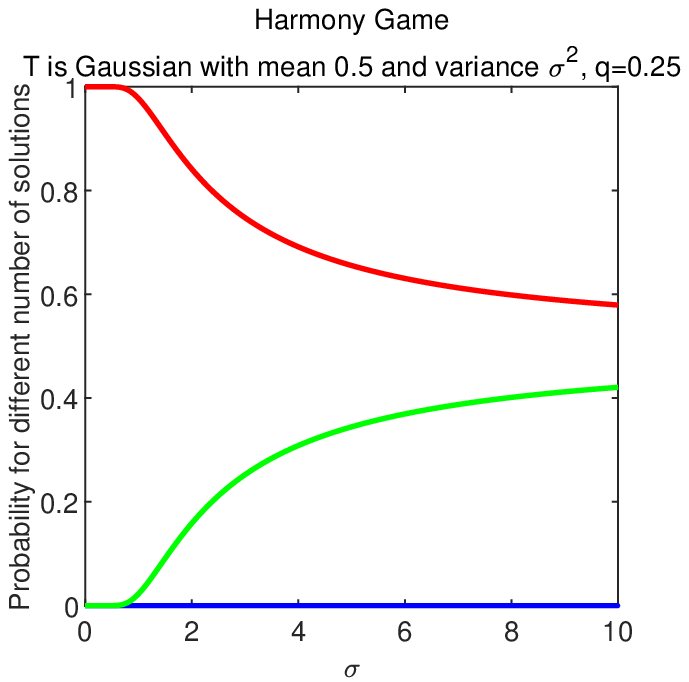}
    }
\caption{The probabilities $p_1 ,p_2 ,p_3$ that the replicator-mutator dynamics has respectively $1,2, 3$ equilibria as functions of $\sigma$, the variance of the random variable T, in PD, SD, SH and H games when $T$ is random and $S$ is fixed.} 
\label{fig: fig6}
\end{figure}

\newpage
\section{$S$ is random}
\label{sec: S random}
In this section, we consider the case where only $S$ is random, assuming that
$$
S=S_0+\varepsilon_S
$$
where $\varepsilon_S \sim N(0,\eta^2)$ and $T=T_0<1$. Then we have
$$
S \sim N(S_0,\eta^2).
$$
As in Section \ref{sec: T random}, we will compute the probability that the replicator-mutator has $k\in\{1,2,3\}$ equilibria, which is the same as the probability that the polynomial $g$ defined in \eqref{eq: g} has $k-1\in\{0,1,2\}$ positive roots. This has been characterized in three corresponding cases (a), (b) and (c) in Section \ref{sec: equilibrium points}.
\subsection{The probability that the replicator-mutator has two equilibria}
In this section, we compute the probability that the replicator-mutator has 2 equilibria, which amounts to computing the probability that the quadratic polynomial $g$ defined in \eqref{eq: g} has 1 positive root. This can happen in two different cases below.
\\ \\
\noindent \textbf{(C1)} $\Delta=0$, $q+a-c<0$. \\ 
Let $\hat{T}:=T-1$ and
 \begin{align} X(\hat{T},S)&:=q+a-c=(1-q)(T_0-1)+qS, \label{S1}\\
Y(\hat{T},S)&:=c-q=q(T_0-1)+(1-q)S.\label{S2}
   \end{align}
 Then, the equation $\Delta=0$ can be rewritten as
   \begin{equation}
       X^2+4qY=-4q^2. \label{S3}
   \end{equation}
Substituting $X$ and $Y$ in (\ref{S1}) and (\ref{S2}) to \eqref{S3}, we get
   $$
   S=\frac{-(1-q)(T_0+1)\pm2\sqrt{T_0(1-2q)}}{q},
   $$
   with $\frac{-1-(T_0-1)q}{1-q}<S<\frac{-qT_0}{1-q}$.
Since $S$ is a continuous random variable, the probability that this case occurs is $0$.
 \\ \\  
   
\noindent \textbf{(C2)} $\Delta>0$ and $c\geq0$.\\

It follows from \eqref{S2} that $c\geq 0$ is equivalent to $Y\geq -q$, which gives 
  \begin{equation*}
       S\geq\frac{-qT_0}{1-q}.
   \end{equation*}
From (\ref{S1}) and (\ref{S2}), we have 
\begin{align*}
      \Delta&=X^2+4qY+4q^2\\
      &=((1-q)(T_0-1)+qS)^2+4q(q(T_0-1)+(1-q)S)+4q^2\\
      &=q^2S^2+2q(1-q)(T_0+1)S+4q^2T_0+(1-q)^2(T_0-1)^2.
      \end{align*}
Thus $\Delta$ can be seen as a quadratic polynomial of $S$ with a leading coefficient $q^2>0$ and its discriminant is given by 
   \begin{equation*}
       \hat{\Delta}=16q^2T_0(1-2q)>0
   \end{equation*}
since for $0<q<\frac{1}{2}$. Thus The equation $\Delta=0$ have two real solutions, $S_1<S_2$, given by
   \begin{equation}
    S_1=\frac{(q-1)(T_0+1)-2\sqrt{T_0-2qT_0}}{q},\quad S_2=\frac{(q-1)(T_0+1)+2\sqrt{T_0-2qT_0}}{q}. \label{S1S2}
     \end{equation}
To proceed, we need to compare $\frac{-qT_0}{1-q}$ with $S_2$. We have
   \begin{align*}
     \frac{\frac{-qT_0}{1-q}}{S_2}&=\frac{-qT_0}{1-q} \frac{q}{(q-1)(T_0+1)+2\sqrt{T_0-2qT_0}}\\
         &=\frac{q^2T_0}{q^2T_0+(q-1)^2+(1-2q)T_0+2\sqrt{T_0-2qT_0}}<1.
   \end{align*}
Since both$\frac{-qT_0}{1-q}$ and $S_2$ are negative, it follows that  $\frac{-qT_0}{1-q}>S_2$. Therefore
\begin{align*}
\mathbb{P}(\Delta>0,~Y\geq -q)=\mathbb{P}(S>\frac{-qT_0}{1-q}).    
\end{align*}
Since $S\sim \mathcal{N}(S_0,\eta^2)$, the probability that $g$ has a unique positive root, which is the probability that the replicator-mutator dynamics have two equilibria, is given by
   \begin{equation}
       \begin{aligned}
           p_2&=\int_{\frac{-qT_0}{1-q}}^{\infty} \frac{1}{\eta\sqrt{2\pi}} e^\frac{-(x-S_0)^2}{2\eta^2} dx\\
           &=\frac{1}{2}-\frac{1}{2}\erf(\frac{\frac{-qT_0}{1-q}-S_0}{\sqrt{2}\eta}).
       \end{aligned}
       \label{S4}
   \end{equation}
\begin{lemma}
As a function of $q$, the probability $p_2$ is increasing in all games. As a consequence,
$$
\frac{1}{2}-\frac{1}{2}\erf(\frac{-S_0}{\sqrt{2}\eta})< p_2< \frac{1}{2}-\frac{1}{2}\erf(\frac{-T_0-S_0}{\sqrt{2}\eta}).
$$
\end{lemma}
\begin{proof}
Let $z:=\frac{\frac{-q T_0}{1-q}-S_0}{\eta\sqrt{2}}$. Then $p_2=\frac{1}{2}-\frac{1}{2}\erf(z)$. Thus, since $T_0>0$, we have
\begin{align*}
\frac{d p_2}{dq}=\frac{dp_2}{dz}\frac{dz}{dq}= \frac{1}{\sqrt{\pi}}e^{-z^2}\frac{T_0}{\eta\sqrt{2}(1-q)^2}>0.
\end{align*}
Thus $p_2$ is increasing as a function of $q$, see Figure \ref{fig:S_1_figure}. As a consequence, we obtain the following lower and upper estimates for $p_2$:
$$
p_2\vert_{q=0}=\frac{1}{2}-\frac{1}{2}\erf(\frac{-S_0}{\sqrt{2}\eta})< p_2<p_e\vert_{q=1/2}= \frac{1}{2}-\frac{1}{2}\erf(\frac{-T_0-S_0}{\sqrt{2}\eta}).$$
\end{proof}
The following example illustrates the above lemma.
\begin{example}
In the SD game, with the specific values $S=0.5, T_0=1.5, \eta=1$ we get
$$
\frac{1}{2}-\frac{1}{2}\erf(-0.5/\sqrt{2})=0.691462<p_2<\frac{1}{2}-\frac{1}{2}\erf(-2/\sqrt{2})=0.97725. 
$$
In the PD game, with $T_0=1.5, S_0=-0.5, \eta=1$, we get
$$
\frac{1}{2}-\frac{1}{2}\erf(0.5/\sqrt{2})=0.308538<p_2<\frac{1}{2}-\frac{1}{2}\erf(-1/\sqrt{2})=0.8413.
$$
In the SH game, with $T_0=0.5, S=-0.5,\eta=1$, we get
$$
\frac{1}{2}-\frac{1}{2}\erf(0.5/\sqrt{2})=0.308538<p_2<\frac{1}{2}-\frac{1}{2}\erf(0)=0.5.
$$
In the H game, with $T_0=S_0=0.5,\eta=1$, 
we get
$$
\frac{1}{2}-\frac{1}{2}\erf(-0.5/\sqrt{2})=0.691462<p_2<\frac{1}{2}-\frac{1}{2}\erf(-1/\sqrt{2})=0.8413.
$$
We notice that in SD and H games, $p_2$ is always dominant.   
\end{example}
\subsection{The probability that the replicator-mutator has three equilibria}
In this section, we compute the probability $p_3$ that the replicator-mutator has three equilibria, that is the probability that $g$ has two positive roots.
We have
\begin{align*}
p_3&=\mathbb{P}(\Delta>0,~ c<0,~ (-q-a+c)>0)=\mathbb{P}(\Delta>0,~ Y<-q,~ X<0)
\\&=\mathbb{P}(\Delta>0,~ S<\frac{(q-1)(T_0-1)}{q},~ S<\frac{-qT_0}{1-q}).
\end{align*}
To proceed, by comparing $\frac{(q-1)(T_0-1)}{q}$ with $\frac{-qT_0}{1-q}$ and with $S_1$, $S_2$, we obtain 2 cases.
\begin{enumerate}[(i)]
    \item When $T_0>\frac{(1-q)^2}{1-2q}$,
    then we have
    $$\frac{-qT_0}{1-q}>S_2>\frac{(q-1)(T_0-1)}{q}>S_1.$$
    Thus, the probability $p_3$ in this case is given by
    $$
    p_3=\mathbb{P}(S<S_1).
    $$
    \item When $T_0<\frac{(1-q)^2}{1-2q}$,
    then we have
    $$\frac{(q-1)(T_0-1)}{q}>\frac{-qT_0}{1-q}>S_2>S_1.$$
    Thus, the probability $p_3$ in this case is given by
    $$
    p_3=\mathbb{P}(S<S_1)+\mathbb{P}(S_2<S<\frac{-qT_0}{1-q}).
    $$
\end{enumerate}
Therefore, since $S\sim \mathcal{N}(S_0,\eta^2)$, the probability of having three equilibria is given by
\begin{equation*}
p_3=
\left\{
   \begin{aligned}
    &\frac{1}{2}+\frac{1}{2}\erf(\frac{S_1-S_0}{\sqrt{2}\eta}), &\text{if } &T_0>\frac{(1-q)^2}{1-2q},\\
    &\frac{1}{2}+\frac{1}{2}\erf(\frac{S_1-S_0}{\sqrt{2}\eta})+\frac{1}{2}\erf(\frac{\frac{-qT_0}{1-q}-S_0}{\sqrt(2}\eta)-\frac{1}{2}\erf(\frac{S_2-S_0}{\sqrt{2}\eta}), &\text{if } &T_0<\frac{(1-q)^2}{1-2q}.
   \end{aligned} 
\right.
\end{equation*}
\noindent In conclusion:
\begin{align*}
p_2&=\frac{1}{2}-\frac{1}{2}\erf(-\frac{\frac{qT_0}{1-q}-S_0}{\sqrt{2}\eta}),\\
p_3&=
\left\{
   \begin{aligned}
    &\frac{1}{2}+\frac{1}{2}\erf(\frac{S_1-S_0}{\sqrt{2}\eta}), &\text{if } &T_0>\frac{(1-q)^2}{1-2q},\\
    &\frac{1}{2}+\frac{1}{2}\erf(\frac{S_1-S_0}{\sqrt{2}\eta})+\frac{1}{2}\erf(\frac{\frac{-qT_0}{1-q}-S_0}{\sqrt(2}\eta)-\frac{1}{2}\erf(\frac{S_2-S_0}{\sqrt{2}\eta}), &\text{if } &T_0<\frac{(1-q)^2}{1-2q}.
   \end{aligned} 
\right.,\\
p_1&=1-p_2-p_3.
  \end{align*}
\subsection{Numerical simulations}
In Figure \ref{fig:S_1_figure} we show the probabilities that PD, SD, SH and H games have $1, 2$ or $3$ equilibria for different values of $q$ using the analytical formula obtained in the previous section (see Theorem 2). Moreover, for validation, these probabilities were calculated by sampling over $10^6$ realizations of $S$ from the normal distribution $\N(S_0,1)$ and then calculating the solutions of ~$g(t)=0$. The values of $T_0$ and $S_0$ are the middle points of the corresponding intervals as in Section \ref{sec: Trandom numerics}. 
It can be seen that the numerical results are clearly in accordance with theoretical ones.
We can also observe that,  $p_2$ is always increasing as a function of $q$, as stated in Theorem 2.

\begin{figure}[htbp]
    \centering
    \subfigure[PD]{
    \includegraphics[scale=1]{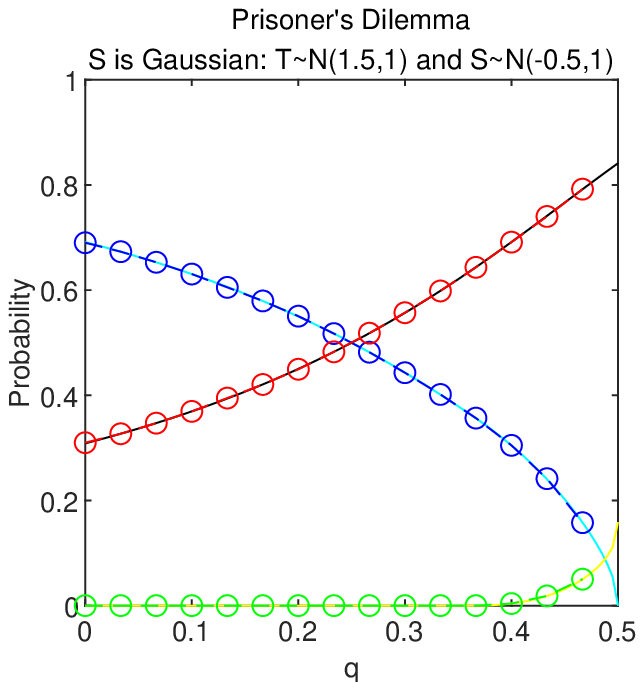}
    }
    \quad
    \subfigure[SD]{
     \includegraphics[scale=1]{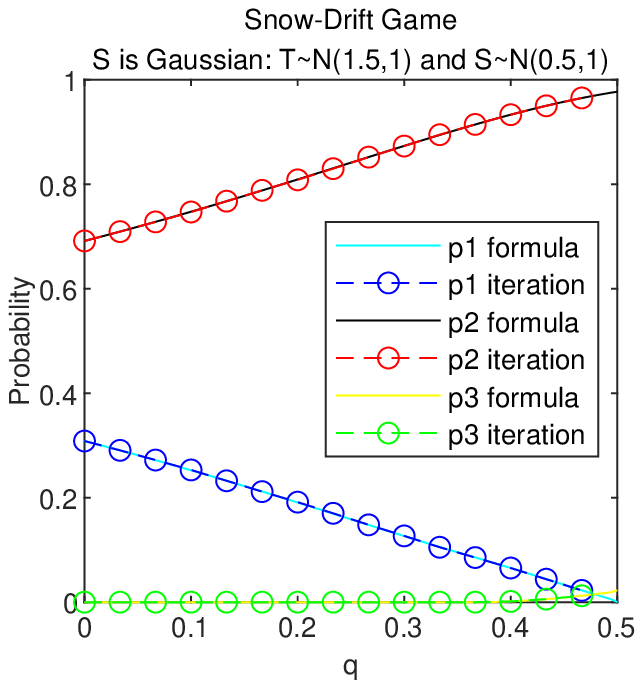}
    }
    \quad
    \subfigure[SH]{
    \includegraphics[scale=1]{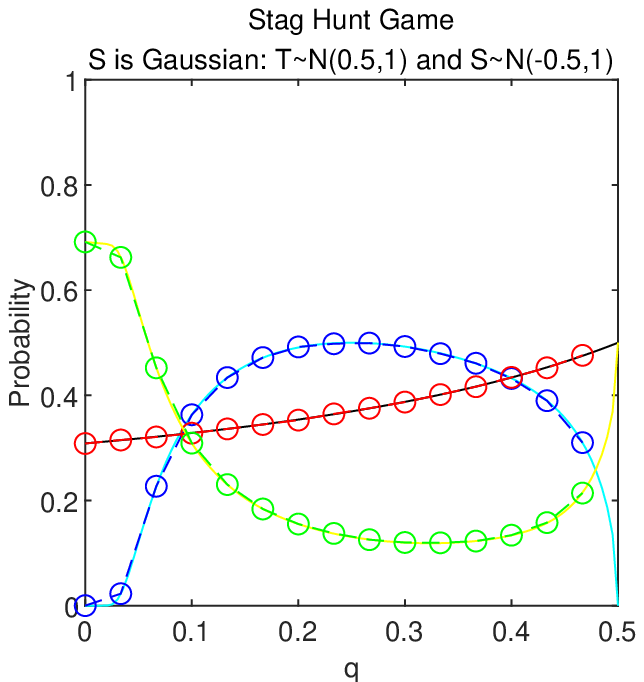}
    }
    \quad
    \subfigure[H]{
     \includegraphics[scale=1]{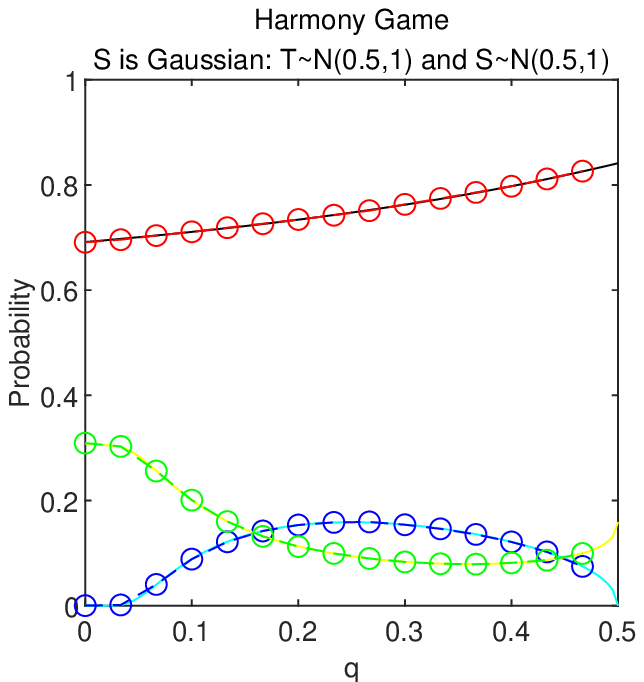}
    }
    \caption{ 
    Probabilities $p_1 ,p_2 ,p_3$ that the replicator-mutator dynamics has respectively $1,2, 3$ equilibria, as functions of the mutation strength $q$ in PD, SD, SH and H games when $S$ is random and $T$ is fixed. The values from analytical analysis and numerical samplings are in accordance.}
    \label{fig:S_1_figure}
\end{figure}

\section{Both $T$ and $S$ are random: numerical investigations}
\label{sec: both T and S are random}
In this section, we numerically compute the probabilities $p_1$, $p_2$ and $p_3$ when both $T$ and $S$ are random. In Figure \ref{fig: fig5}, we show  their values obtained from averaging over $10^6$ random samples of $T$ and $S$ with the corresponding distribution in each game. We observe that $p_2$ tends to increase in all games, while $p_1$ and $p_3$ exhibit more complex behaviours. We aim to study this more complex case analytically in future work. 

\begin{figure}[htbp]
\centering
\subfigure[PD]{
    \includegraphics[scale=1]{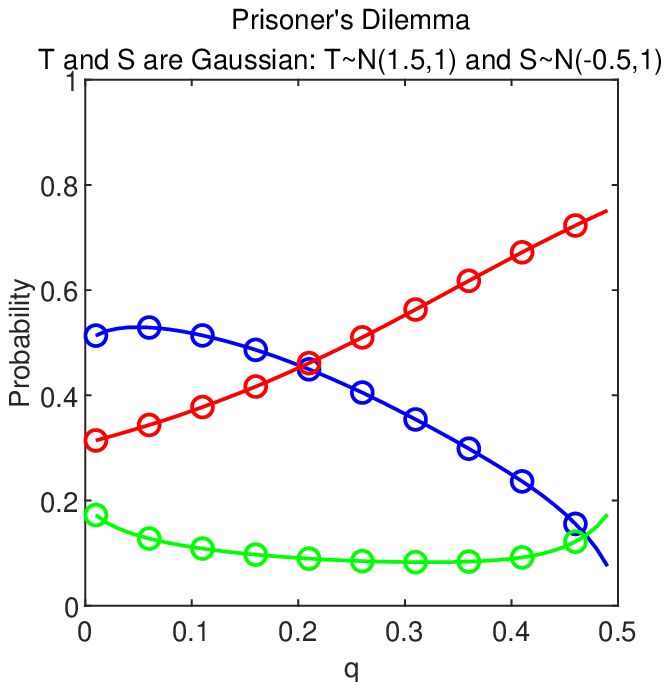}
    }
    \quad
    \subfigure[SD]{
     \includegraphics[scale=1]{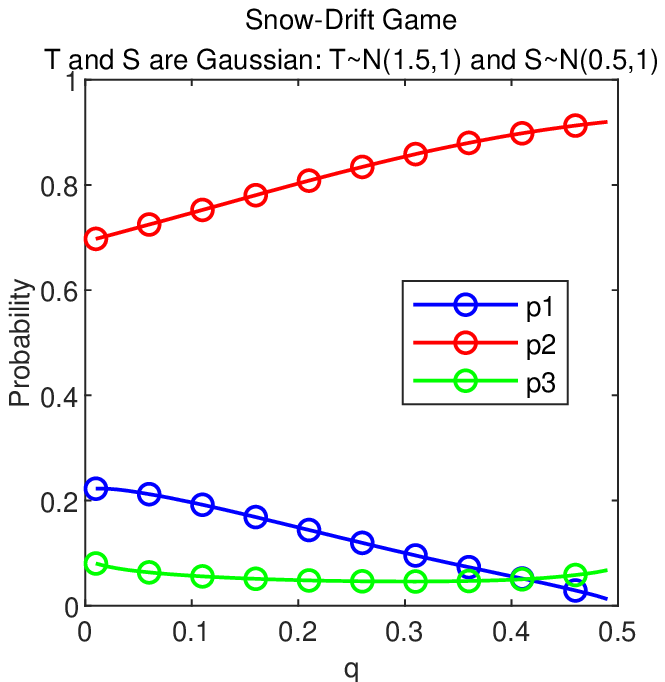}
    }
    \quad
    \subfigure[SH]{
    \includegraphics[scale=1]{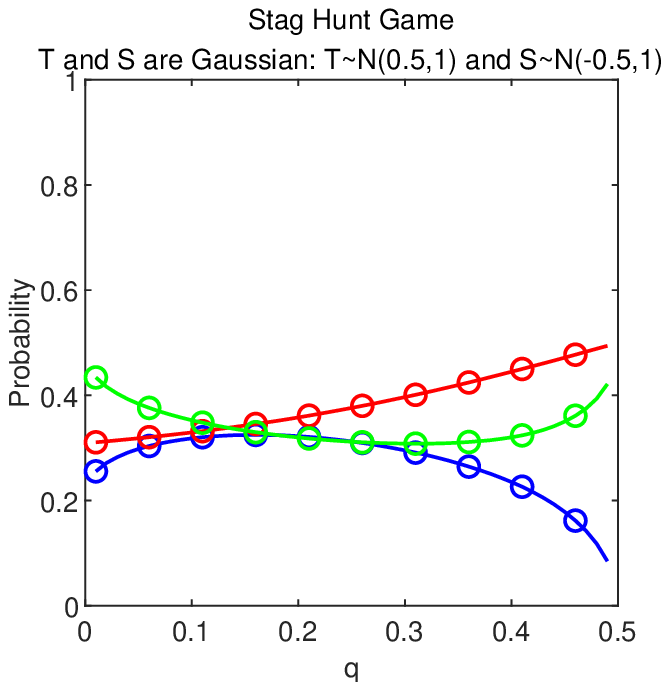}
    }
    \quad
    \subfigure[H]{
     \includegraphics[scale=1]{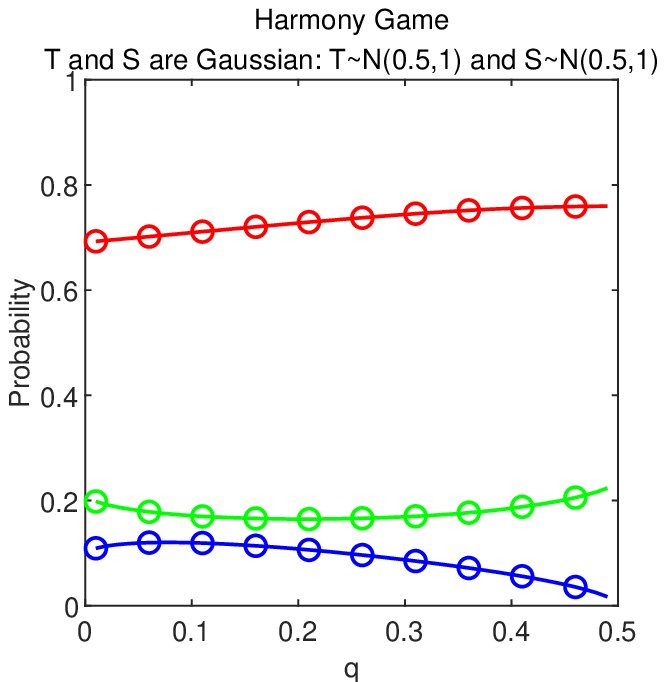}
    }
\caption{
Probabilities $p_1, p_2, p_3$ that the replicator-mutator dynamics has respectively $1, 2, 3$ equilibria, as functions of the mutation strength $q$ in PD, SD, SH and H games when both $T$ and $S$ are random. The values are numerically obtained from $10^6$ samplings.}
\label{fig: fig5}
\end{figure}
\section{Summary and Outlook}
\label{sec: summary}
In this paper, we have studied pair-wise social dilemmas where the payoff entries are random variables. The randomness is necessary to capture the uncertainty that is unavoidable in practical applications, which may come from different sources, both subjective and objective, such as lack of data, fluctuating environment as well as human estimate errors. We have focused on four important social dilemma games, namely Prisoner's Dilemma, the Snow-Drift game, the Stag-Hunt game and the Harmony game. For each game, we have analytically computed, and numerically validated, the probability that the replicator-mutator dynamics has a certain number of equilibria, studying their qualitative behaviour as a function of the mutation rate. 
Our results have clearly shown that the mutation rate and randomness from the payoff matrix have a  strong impact on the equilibrium outcomes. Thus, our analysis has provided novel theoretical contributions to the understanding of the impact of uncertainty on the behavioural diversity in a complex dynamical system.   

Here we have assumed that the payoff entries are standard normal distributions; however, from formulas such as Equation \eqref{eq: p2T}, our results can be easily extended to other distributions. One natural and challenging problem for future work is to generalize the equilibrium analysis of the present work to multi-player and multi-strategy games where the payoff entries satisfy more complex conditions. Another direction is to trajectorially characterize statistical properties of the full dynamical systems.

\end{document}